\documentclass[11pt,a4paper]{article}

\usepackage{amsmath,amssymb,amsthm,graphicx,braket,multirow,authblk,amsthm,dcolumn,latexsym,mathtools,cite,cancel}
\usepackage[normalem]{ulem}
\usepackage{subcaption}
\usepackage[a4paper]{geometry}
\usepackage{booktabs}

\usepackage{amsmath,amssymb,amsthm,mathtools}
\usepackage{hyperref}
\usepackage{enumitem}
\usepackage{bm}

\usepackage{tikz}
\usepackage{pgfplots}
\pgfplotsset{compat=1.18}
\newcommand{\Md}{M_d(\mathbb{C})}
\def\<{\langle}
\def\>{\rangle}
\def\oper{{\mathchoice{\rm 1\mskip-4mu l}{\rm 1\mskip-4mu l}
		{\rm 1\mskip-4.5mu l}{\rm 1\mskip-5mu l}}}
\newcommand{\tr}{\operatorname{Tr}}

\newtheorem{theorem}{Theorem}[section]
\newtheorem{proposition}[theorem]{Proposition}
\newtheorem{remark}[theorem]{Remark}
\newtheorem{lemma}[theorem]{Lemma}

\theoremstyle{definition}

\newtheorem{example}{Example}[section]
\def\<{\langle}
\def\>{\rangle}
\def\D{\Delta}
\newcommand{\ip}[2]{\left\langle #1,#2\right\rangle}

\DeclareMathOperator{\re}{Re}

\newcommand{\Tr}{\operatorname{Tr}}

\newcommand{\C}{\mathbb{C}}

\date{}

\begin{document}


\title{{\bf Interpolating between positive, Schwarz, and completely positive evolution for $d$-level systems  }} 


\author{Dariusz Chru{\'s}ci{\'n}ski$^1$ and Farrukh Mukhamedov$^2$}
\affil{$^1$Institute of Physics, Faculty of Physics, Astronomy and Informatics, \\Nicolaus Copernicus University, Grudzi\c{a}dzka 5/7, 87--100 Toru{\'n}, Poland\\
$^2$Department of Mathematical Sciences\\
College of Science, The United Arab Emirates University\\
P.O. Box, 15551, Al Ain, Abu Dhabi, UAE}

\maketitle

\begin{abstract}

We study a class of quantum dynamical maps for $d$-level systems that interpolate between positive, Schwarz, and completely positive evolutions. Our approach is based on a geometric analysis of the parameter space, which reveals the structure of regions corresponding to different positivity classes and their boundaries. We show that dynamical trajectories naturally move across these regions, providing a clear geometric interpretation of transitions between Markovian and non-Markovian regimes. It is shown that within presented class the evolution becomes eventually entanglement breaking.  This analysis highlights the role of divisibility and 
eternally non-Markovian evolution. 
\end{abstract}

\section{Introduction}

Markovian dynamical semigroups are governed by the celebrated
Gorini--Kossakowski-Lindblad--Sudarshan (GKLS) master equation
\cite{L,GKS,Alicki} $\dot\rho_t = \mathcal{L}(\rho_t)$, where the
corresponding Markovian generator has the following form
\begin{equation}
  \mathcal{L}(\rho) = -i[H,\rho] + \sum_k \gamma_k
  \left(L_k \rho L_k^\dagger - \tfrac{1}{2}\{L_k^\dagger L_k,\rho\}\right),
  \label{GKLS}
\end{equation}
where $H$ denotes an effective system Hamiltonian, $L_k$ are noise
(jump) operators, and $\gamma_k > 0$ are positive transition rates
(in what follows we set $\hbar=1$ and denote the anti-commutator
$\{A,B\} = AB+BA$). This gives rise to the general representation for
the generator of the Markovian semigroup
$\{\Lambda_t = e^{t\mathcal{L}}\}_{t\ge 0}$ of completely positive
trace-preserving maps (CPTP) \cite{Paulsen,STORMER,Bhatia}. Solutions of the GKLS master
equation define very good approximations to the evolution of many real
systems, provided the system--environment interaction is sufficiently
weak and there is a large enough separation of timescales between
the system and environment \cite{Open1,Open2,Open3,Open4}.

In this paper we generalize a paradigmatic qubit generator
\begin{equation}
  \mathcal{L}(\rho) = -i\tfrac{\omega}{2}[\sigma_z,\rho]
  + \kappa\!\left((\sigma_+\rho\sigma_- + \sigma_-\rho\sigma_+ - \rho)
  + \tfrac{\nu}{2}(\sigma_z\rho\sigma_z - \rho)\right),
  \label{qubit-gen}
\end{equation}
where $\sigma_+$ and $\sigma_- = \sigma_+^\dagger$ are standard
raising and lowering qubit operators, for arbitrary $d>2$. Clearly,
\eqref{qubit-gen} defines a legitimate GKLS generator if $\kappa,\nu\ge 0$.
Interestingly, even if $\nu < 0$ the dynamical map $e^{t\mathcal{L}}$
may be positive or even may satisfy an operator Schwarz inequality \cite{Paulsen,STORMER,Bhatia}
\begin{equation}
  \Lambda_t(XX^\dagger) \ge \Lambda_t(X)\Lambda_t(X^\dagger),
  \label{schwarz}
\end{equation}
for all $X\in M_d(\mathbb{C})$. Indeed, one shows that $\Lambda_t = e^{t\mathcal{L}}$ is
positive for all $t\ge 0$ iff $\kappa > 0$, $\nu \ge -1$, and it satisfies
\eqref{schwarz} iff $\kappa > 0$, $\nu \ge -\tfrac{1}{2}$ \cite{KS,Sergey,Erica,Theret,Theret2,CKM,PR}. Interestingly, the Schwarz
inequality \eqref{schwarz} plays an important role in the analysis of
asymptotic properties of quantum dynamics \cite{Facchi} and in the study of
quantum channels in general \cite{Wolf} (cf. also \cite{Alex}). Why to consider dynamical maps beyond standard completely positive scenario?  Recall, that the GKLS equation is not the only master equation
used for describing quantum dynamics. An important alternative, widely employed in
magnetic resonance, quantum optics, and condensed matter, is the
Bloch--Redfield equation \cite{Bloch-1,Bloch-2}, derived from the Born--Markov
approximation without performing the additional secular (rotating-wave)
approximation required to reach GKLS form. The secular approximation
eliminates rapidly oscillating terms and thereby enforces complete
positivity. The Born-Redfield equation usually yields more accurate coherence dynamics in multilevel systems. 
However, the price paid is that the  Born-Redfield generator is \emph{not} of GKLS form, and
the corresponding dynamical maps are not guaranteed to be completely
positive: for sufficiently strong coupling or at short times it can produce transient violations of complete positivity, and in certain parameter regimes even of positivity  \cite{Open1,Whitney,Fabio}. Another motivation to go beyond completely positive maps is the analysis of non-Markovian dynamics and in partucular P-divisible or Schwarz-divisible evolutions. The dynamical map $\{\Lambda_t\}_{t \geq 0}$ is completely positive but the family of propagators $\{\Lambda_{t,s}\}_{t\geq s}$ in general fails to be completely positive. One calls  $\{\Lambda_t\}_{t \geq 0}$ P-divisible if all $\Lambda_{t,s}$ are positive and Schwarz-divisible if all propagators in the Heisenberg picture satisfy operator Schwarz inequality \cite{KS}.

In this paper we analyze the $d$-level generalization of
\eqref{qubit-gen}. We provide conditions which guarantee that the
corresponding semigroup $e^{t\mathcal{L}}$ is (i) completely positive, (ii) positive, or (iii) satisfies
the Schwarz inequality. This way we provide a new class of Schwarz maps
in the matrix algebra $M_d(\mathbb{C})$. Recently~\cite{Bihalan} such maps
were analyzed in the qubit case. Characterization beyond $d=2$ is
notoriously difficult. Moreover, we analyze time-dependent generators
and analyze non-Markovian evolution, i.e.\ evolution characterized by
dynamical maps $\Lambda_t$ which being CPTP are not CP-divisible (i.e. propagators $\Lambda_{t,s}$ are not completely positive). We
characterize a class of admissible dynamical maps which are P-divisible
and Schwarz-divisible. In particular we construct an example of
$d$-level \emph{eternally} non-Markovian dynamics which provides a
direct generalization of the well known qubit evolution~\cite{Erica}. 
Interestingly, our analysis has a clear geometric interpretation. In particular,  eternally non-Markovian dynamics belongs to the border of the convex set of completely positive maps. It is shown that the class of dynamics studied in this paper is eventually entanglement breaking \cite{EB,EB1,EB2}. For a semigroup (completely positive or just positive) evolutions becomes entanglement breaking after a finite time. However, we analyze a non-Markovian scenario for which it becomes entanglement breaking only asymptotically.

\section{Preliminaries: conditionally positive and dissipative generators}

Let $\Md$ denote the space of $d \times d$ complex matrices. A linear map $\Phi : \Md \to \Md$ is called positive if for any $X \geq 0$ one has $\Phi(X) \geq 0$. It is $k$-positive if the extended map

\[   {\rm id}_k \otimes \Phi : \mathcal{M}(\Md) \to  \mathcal{M}(\Md) , \]
is positive. Finally, $\Phi$ is completely positive (CP) if it is positive for all $k=1,2,\ldots$. Actually, $\Phi : \Md \to \Md$ is CP if and only if it is $d$-positive. Given $\Phi$ one defines its Hilbert-Schmidt adjoint $\Phi^\ddag : \Md \to \Md$  via

\begin{equation}
    (\Phi^\ddag(X),Y)_{\rm HS} := (X,\Phi(Y))_{\rm HS} ,   
\end{equation}
where $(X,Y)_{\rm HS} = \Tr(X^\dagger Y)$. Now, $\Phi^\ddag$ is $k$-positive if and only if $\Phi$ is $k$-positive. and $\Phi^\ddag$ is unital (trace-preserving) if and only if $\Phi$ is trace preserving (unital). 
A unital linear map $\Phi : \Md \to \Md$, that is, $\Phi(\oper_d) = \oper_d$,  is called a Schwarz map if

\begin{equation}\label{KS}
  \Phi(X^\dagger X) \geq \Phi(X)^\dagger \Phi(X)  ,
\end{equation}
for all $X \in \Md$ \cite{Paulsen,STORMER,KADISON,Bhatia}. It was shown by Kadison \cite{Kadison1,Kadison2} that if $\Phi$ is a positive unital map then it satisfies celebrated Kadison inequality

\begin{equation}\label{K}
  \Phi(X^2) \geq \Phi(X)^2   ,
\end{equation}
for all $X^\dagger=X$. However, not all positive unital maps satisfy (\ref{KS}). The simplest example is provided by the transposition map.

Consider now a semigroup $\{\Lambda_t = e^{t\mathcal{L}}\}_{t \geq 0}$, where $\mathcal{L} : \Md \to \Md$ stands for the corresponding generator. $\Lambda_t$ is trace-preserving if and only if $\Tr \Lambda_t(\rho)=0$ for any $\rho \in \Md$. Equivalently, $\mathcal{L}^\ddag(\oper_d)=0$.

\begin{proposition}[\cite{Kos72}] $\Lambda_t=e^{t \mathcal{L}}$ is a semigroup of  positive maps for $t\geq 0$ if and only if

\begin{equation} \label{CPP}
 \<y|\mathcal{L}(|x\>\<x|)|y\> \geq 0 ,
\end{equation}
for any pair of orthogonal vectors $x,y \in \mathbb{C}^d$. 
\end{proposition}
Similarly,

\begin{proposition} $\Lambda_t=e^{t \mathcal{L}}$ is a semigroup of  $k$-positive maps for $t\geq 0$ if and only if

\begin{equation} \label{kCPP}
 \<\tilde{y}|[{\rm id}_k \otimes \mathcal{L}](|\tilde{x}\>\<\tilde{x}|)|\tilde{y}\> \geq 0 ,
\end{equation}
for any pair of orthogonal vectors $\tilde{x},\tilde{y} \in \mathbb{C}^k \otimes \mathbb{C}^d$. 
\end{proposition}
A map $\mathcal{L} : \Md \to \Md$ satisfying (\ref{CPP}) is called conditionally positive \cite{Evans-1979,Evans-1977}.  

\begin{proposition} $\mathcal{L}$ generates a semigroup of CP maps iff it is conditionally completely positive, i.e.

\begin{equation}  \label{CCP}
 P^+_\perp C_\mathcal{L} P^+_\perp \geq 0 ,
\end{equation}
where $C_\mathcal{L}$ stands for the Choi matrix of $\mathcal{L}$, and $P^+$ denotes a projector onto maximally entangled state in $\mathbb{C}^d \otimes \mathbb{C}^d$ ($P^+_\perp = \oper_d \otimes \oper_d - P^+$). \end{proposition}
Lindblad \cite{L} provided the following  condition for the generator $\mathcal{L}$ for which  $\{\Lambda^\ddag_t\}_{t\geq 0}  $ is a semigroup of Schwarz maps.

\begin{proposition}[\cite{L}]  \label{PRO-L} $\mathcal{L}^\ddag$ generates a semigroup $\{\Lambda_t^\ddag\}_{t\geq 0}$ of unital Schwarz maps if and only if $\mathcal{L}^\ddag(\oper)=0$ and

\begin{equation}\label{L!}
  \mathcal{L}^\ddag(X^\dagger X) \geq \mathcal{L}^\ddag(X^\dagger) X + X^\dagger \mathcal{L}^\ddag(X) ,
\end{equation}
for all $X  \in \Md$.
\end{proposition}
A map $\mathcal{L} : \Md \to \Md$ satisfying (\ref{L!}) is called {\em dissipative} \cite{L} and it is called $k$-dissipative if ${\rm id}_k \otimes \mathcal{L}$ is dissipative. Finally, $\mathcal{L}$ is completely dissipative if it is $k$-dissipative for all $k$. One proves \cite{L,Evans-1979,Evans-1977}
that $\mathcal{L}$ is completely dissipative if and only if it is  conditionally completely positive. 
Interestingly, restricting (\ref{L!}) to Hermitian operators one arrives at \cite{Evans-1979,Evans-1977}

\begin{proposition} \label{P-KS} $\mathcal{L}^\ddag$ generates a semigroup $\{\Lambda_t^\ddag\}_{t\geq 0}$ of positive unital maps if and only if $\mathcal{L}^\ddag(\oper)=0$ and

\begin{equation}\label{KS-LH}
  \mathcal{L}^\ddag(X^2) \geq \mathcal{L}^\ddag(X) X + X \mathcal{L}^\ddag(X) ,
\end{equation}
for all $X =X^\dagger \in \Md$.
\end{proposition}

Note that a generator of a unital Hermiticity-preserving semigroup has the following representation \cite{L,GKS,Alicki}:
\begin{equation}\label{HPGen}
	\mathcal{L}^\ddag(X) = i[H,X] + \Phi^\ddag(X) - \frac 12 \{\Phi^\ddag(\oper),X\} ,
\end{equation}
{where $H = H^\dagger$ and $\Phi$ is a Hermiticity-preserving map. Using \eqref{L!}}, one finds \cite{L} that $\mathcal{L}^\ddag$ gives rise to a semigroup of Schwarz maps if and only if

\begin{equation}\label{Phi!}
  \Phi^\ddag(X^\dagger X) \geq \Phi^\ddag(X^\dagger) X + X^\dagger \Phi^\ddag(X) - X^\dagger {\Phi^\ddag(\oper)}X,
\end{equation}
for all $X  \in \mathcal{M}_n$. In particular, if $\Phi$ is trace preserving (equivalently, if $\Phi^\ddag$ is unital)  the above condition reduces to 

\begin{equation}\label{Phi!1}
  \Phi^\ddag(X^\dagger X) \geq \Phi^\ddag(X^\dagger) X + X^\dagger \Phi^\ddag(X) - X^\dagger X,
\end{equation}
Actually, it is sufficient to check (\ref{Phi!}) (or (\ref{Phi!1})) for traceless operators \cite{CKM}.

\section{A class of qudit Markovian generators}

Let $\{|1\>,\ldots,|d\>\}$ be a computational basis in $\mathbb{C}^d$ and $E_{ij} = |i\>\<j|$ be an orthonormal basis in $\Md$ (w.r.t. Hilbert-Schmidt inner product).  Consider a generator defined by 

\begin{equation}  \label{L}
\mathcal{L}(\rho)
=
-i[H,\rho]
+
\kappa \left[ 
\left(\sum_{i\neq j} E_{ij}\rho E_{ji} - (d-1)\rho\right)
+
 \frac{\nu}{d} \left( \sum_{k=1}^{d-1} Z^k \rho Z^{*k} - (d-1) \rho \right) 
\right] ,
\end{equation}
where $H = \sum_k h_k E_{kk}$, and  

$$ Z= \sum_{\ell=1}^{d} e^{2 \pi i \ell/d} E_{\ell\ell} , $$ 
is a diagonal unitary matrix (generalization of $\sigma_z$).  It is clear that for $d=2$ it reduces to (\ref{qubit-gen}). 
Interestingly conditions for positivity and complete positivity are universal and do not depend on the dimension. 
\begin{proposition} A map $\mathcal{L} : \Md \to \Md$ defined in (\ref{L}) is

\begin{enumerate}
\item conditionally completely positive iff $\nu \geq 0 $,
    \item conditionally positive iff $\nu\geq -1$,
\end{enumerate}
    
\end{proposition}
\begin{proof}

\begin{enumerate}
    \item $\mathcal{L}$ is conditionally completely positive iff its Choi matrix satisfies (\ref{CCP}).
One finds

\[ \< P^+_\perp|C_\mathcal{L}|P^+_\perp\> =  \kappa \Big( \oper_d \otimes \oper_d - (1-\nu) \sum_{k=1}^d P_k \otimes P_k \Big)  ,\]
and hence $\< P^+_\perp|C_\mathcal{L}|P^+_\perp\> \geq 0$ if and only if $1-\nu \leq 1$, that is,  $\nu \geq 0$. 

\item  $\mathcal{L}$ is conditionally positive iff its Choi matrix satisfies (\ref{CPP}). Let $x=\sum_k x_k e_k$ and $y = \sum_k y_k e_k$ be two orthogonal normalized vectors from $\mathbb{C}^d$. One finds

\[ \< y|\mathcal{L}(|x\>\<x|)|y\> = \kappa \Big( 1 - (1-\nu) \sum_{k=1}^d |x_k|^2 |y_k|^2 \Big) .  \]

\begin{lemma}   \label{L1}
For any normalized and orthogonal vectors \(x,y\in \mathbb{C}^d\)  one has
\begin{equation}
    \sum_{i=1}^d |x_i|^2\,|y_i|^2 \le \frac12.
\end{equation}
\end{lemma}
For the proof see Appendix. The above Lemma immediately implies  that $\< y|\mathcal{L}(|x\>\<x|)|y\> \geq 0$ if and only if $\nu \geq -1$.
\end{enumerate}
\end{proof}

\begin{proposition} \label{P-MAIN} A map $\mathcal{L} : \Md \to \Md$ is dissipative, i.e. satisfies (\ref{L!}),  iff $\nu \geq - \frac{d}{d+2}$.    
\end{proposition}
For the proof see Appendix. Summarizing, 

\[     e^{t \mathcal{L}}  \leftrightarrow \left\{  \begin{array}{ll} {\rm semigroup\ of\ positive\ maps} &  ; \ \nu \geq - 1 \\
{\rm semigroup\ of\ Schwarz\ maps} & ;\ \nu \geq - \frac{d}{d+2} \\
{\rm semigroup\ of\ completely\ positive\ maps} & ; \ \nu \geq 0  \end{array} \right. \ \ .
\]
For $d=2$ we recover the result from \cite{CKM}: $e^{t \mathcal{L}}$ defines a semigroup of Schwarz maps iff $\nu \geq -\frac  12$.

\section{Markovian semigroups}

The dynamical map $\Lambda_t = e^{t \mathcal{L}} $ has the following form

\begin{equation}
    \Lambda_t = \mathcal{U}_t \circ \Phi_t ,
\end{equation}
where $\mathcal{U}_t(\rho) = e^{-i H t}\rho \, e^{i Ht}$, and


\begin{equation}
    \Phi_t(\rho) = e^{-\kappa(d-1+\nu) t} \,\rho + \Big( e^{-\kappa d t} - e^{-\kappa(d-1+\nu) t} \Big) \Delta(\rho) + \frac 1d \Big( 1 - e^{-\kappa d t}\Big)\oper_d \tr\rho .
\end{equation}
It is clear that $\Lambda_t$ is relaxing to the maximally mixed state $\oper_d/d$, i.e. for any initial state $\rho$ one finds

\begin{equation}
    \lim_{t \to \infty} \Lambda_t(\rho) = \frac{\oper_d}{d} . 
\end{equation}
Note, that for any $t$ the map $ \Phi_t$ belongs to the family

\begin{equation} \label{Phi-map}
 \Phi_{\alpha,\beta}   =   (1-\alpha - \beta) {\rm id} + \alpha\, \tau_0 + \beta\, \Delta ,    
\end{equation}
where $\tau_0(\rho) = \frac 1d \oper_d \tr \rho$. One proves \cite{BBC} the following 

\begin{proposition} \label{PCP} A linear map (\ref{Phi-map})

\begin{enumerate}
    \item is positive  iff

\begin{equation}   \label{1P}
    0 \leq \alpha \leq \frac{d}{d-1} \ , \quad - \frac{2\alpha}{d} \leq \beta \leq \frac{d}{d-1} - \alpha , 
\end{equation}

\item is completely positive iff

\begin{equation}  \label{CP}
    0 \leq \alpha \leq \frac{d}{d-1} \ , \quad - \frac{\alpha}{d} \leq \beta \leq \frac{d}{d-1} - \frac{d+1}{d}\, \alpha .
\end{equation}
    
\end{enumerate} 
    
\end{proposition}
Moreover, one has

\begin{proposition} \label{Pro-EB} A linear map (\ref{Phi-map}) defines an entanglement-breaking (EB) channel iff 

\begin{equation}
 \max \left( -\frac{\alpha}{d}, 1 - \alpha - \frac{\alpha}{d} \right) \leq \beta \leq \min \left( \frac{d}{d - 1} - \frac{d + 1}{d} \alpha, 1 - \alpha + \frac{\alpha}{d} \right) ,
\end{equation}
for any $0 \leq \alpha \leq \frac{d}{d-1} $.
\end{proposition}
For the proof see Appendix. All these regions, i.e. positive maps, completely positive maps, and entanglement breaking maps, are presented in Fig. \ref{Fig-1}:

\[ \mbox{Positive maps} = {\rm conv} \{{\rm id}, \Phi, \mathcal{R}, \mathcal{P} \} ,  \]
where

\begin{equation}
    \mathcal{R}(X) = \frac{1}{d-1} \left(\oper_d \tr X -X \right), \quad \mathcal{P}(X) = \frac{1}{d-1} \left( \oper_d \tr X + X -2 \Delta(X) \right),
\end{equation}
are positive and not CP ($\mathcal{R}$ is a well known reduction map). The CP map $\Phi$ corresponding to $(0,\frac{d}{d-1})$ is defined by

\begin{eqnarray}
     \Phi(X) = \frac{1}{d-1} \Big( d\,\Delta(X) - X \Big) .
\end{eqnarray}
Finally, four EB maps $\mathcal{E}_k$ corresponding to $(\alpha_k,\beta_k)$

\[  \mathcal{E}_1 \leftrightarrow (0,1) \ , \ 
 \mathcal{E}_2 \leftrightarrow \left(\frac{d}{2(d-1)},\frac 12\right) \ , \ 
  \mathcal{E}_3 \leftrightarrow \left(\frac{d}{d-1},-\frac{1}{d-1}\right) \ , \ 
  \mathcal{E}_4 \leftrightarrow \left(1,-\frac{1}{d}\right)
 \]
read as follows:

\begin{eqnarray}
     \mathcal{E}_1(X) &=& \Delta(X) \ , \nonumber \\
     \mathcal{E}_2(X) &=& \frac 12 \Big(  \frac{1}{d-1} \oper_d \Tr X + \Delta(X) - \frac{1}{d-1} \, X \Big)  \ , \nonumber \\
     \mathcal{E}_3(X) &=& \frac{1}{d-1}\Big( \oper_d \Tr X - \Delta(X) \Big)  \ ,  \\
     \mathcal{E}_4(X) &=& \frac 1d \Big( X + \oper_d \Tr X - \Delta(X)\Big)  \ . \nonumber 
\end{eqnarray}

\begin{figure}[h!]
    \centering
    \includegraphics[width=10cm]{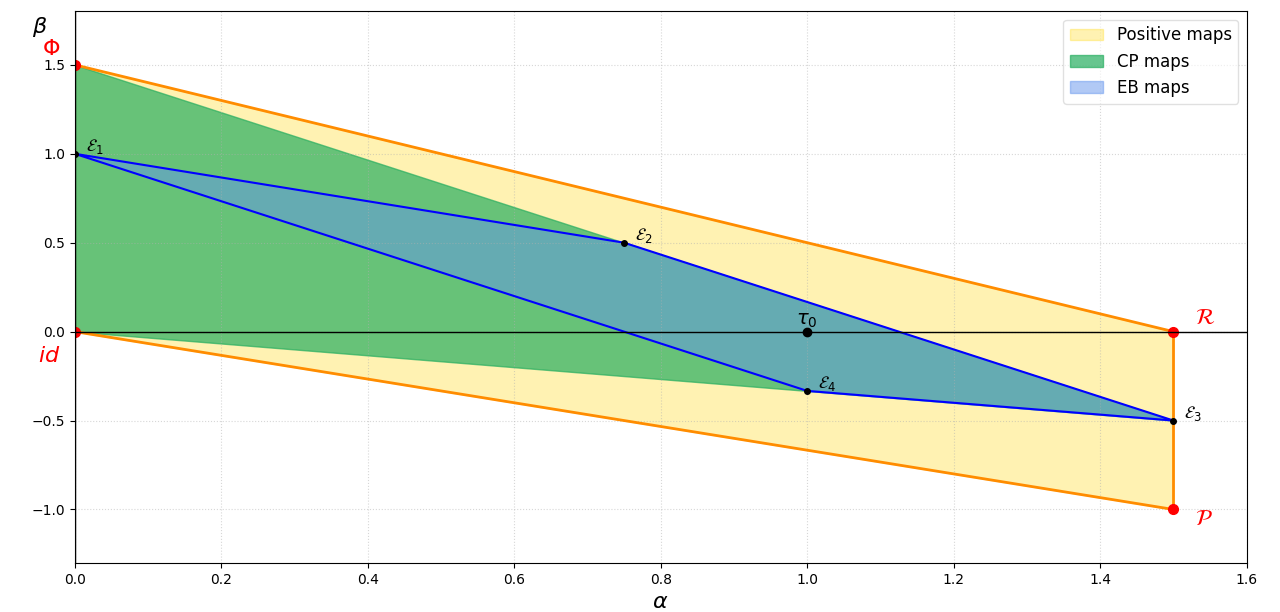} 
    \caption{Regions of positive maps, completely positive maps, and entanglement breaking maps for $d=3$.
  }  \label{Fig-1}
\end{figure}
One easily compute the are of the corresponding regions:
\[   {\rm Area}[P] = \frac{d(d+2)}{2(d-1)^2}\ , \ \ {\rm Area}[CP] = \frac{d^2}{2(d-1)^2}\ , \ \ {\rm Area}[EB] = \frac{3d -2}{2d(d-1)}\ . \ \ \]
Hence, for $d \gg 1$ one finds 

\[ \frac{{\rm Area}[P]}{{\rm Area}[CP]} \to 1 \ ,  \ \ \  \frac{{\rm Area}[EB]}{{\rm Area}[CP]} \to 0 . \]
One finds 

\begin{equation}
    \Phi_t = [1-\alpha(t) - \beta(t)] {\rm id} + \alpha(t) \tau_0 + \beta(t) \Delta ,
\end{equation}
where  the time dependent parameters

\begin{equation}  \label{ab}
   \alpha(t) = 1 - e^{-\kappa d t} \ , \quad \beta(t) =  e^{-\kappa d t} - e^{-\kappa(d-1 +\nu) t} = e^{-\kappa d t}(1 - e^{\kappa (1-\nu) t}).  
\end{equation}
Note, that for any $t \geq 0$ one has $\alpha(t) \geq 0$ and $\beta(t) \leq 0$.

\begin{figure}[h!]
    \centering
    \includegraphics[width=10cm]{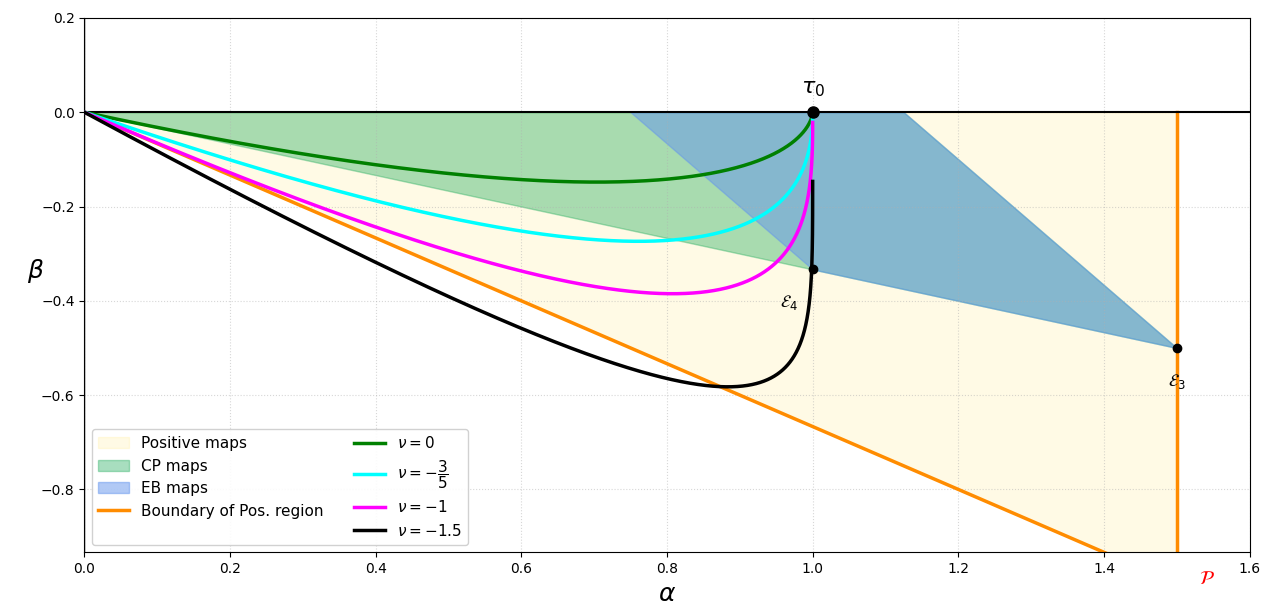} 
    \caption{Trajectories {$(\alpha(t),\beta(t))$} corresponding to various value of the parameter $\nu$.  $d=3$,  $\kappa=1$.}
    \label{Fig-2}
\end{figure}

Fig. \ref{Fig-2} presents trajectories on the $(\alpha,\beta)$ plane for various value of the parameter $\nu$. Note, that for $\nu=0$, the green trajectory is tangent to CP region at $t=0$. Similarly,  for $\nu=-1$, the cyjan trajectory is tangent to positivity region at $t=0$. Taking $\nu = -\dfrac{d}{d+2}$ one finds the blue trajectory tangential to the Schwarz region at $t=0$. Note, that for $\nu < -1$ the trajectory leaves the positivity region for $t>0$ but then return to positivity region at $t_P> 0$, enters CP region at $t_{CP}> t_P$, and eventually enters EB region  at $t=t_{EB}> t_{CP}$.

\section{Relaxation rates}

One easily finds the spectrum of $\mathcal{L}$. Observe that

\begin{equation}
    \mathcal{L}(E_{k\ell}) = [-i(h_k - h_\ell) - \kappa(d-1+\nu)] E_{k\ell} ,
\end{equation}
and the remaining eigenvalues correspond to the spectrum of the matrix $K$:

\[ K = \kappa (\mathbb{J}_d - d \oper_d) ,\]
where $(\mathbb{J}_d)_{k\ell}=1$. Hence, there are additionally $d$ eigenvalues:

\[  \lambda_0=0,\ \lambda_1=\ldots=\lambda_{d-1}= - \kappa d . \]
Hence, the corresponding relaxation rates read as follows:

\[  \Gamma_\ell = \kappa d \ , \ \ell=1,\ldots,d-1 ,\]
and
\[  \Gamma_{ij} = \kappa(d-1+\nu) \ , \ \ i \neq j . \]
The total rates reads

\begin{equation}
    \Gamma = \sum_k \Gamma_k + \sum_{i\neq j} \Gamma_{ij} = \kappa d(d-1)(d+\nu) . 
\end{equation}
It was proved recently \cite{ROPP} (cf. also \cite{JPA-GP,PRL}) that 

\begin{equation}   \label{GG}
    \Gamma_{\rm max} \leq c_d \Gamma ,
\end{equation}
where

\begin{equation}
    c_d = \left\{ \begin{array}{cl} 1 & {\rm semigroup\ of\ positive\ maps} \\ \frac{2}{d+1} & {\rm semigroup\ of\ Schwarz\ maps} \\ \frac 1d & {\rm semigroup\ of} \ k{\rm -positive\ maps}\ \ (k\geq 2) \end{array} \right. . 
\end{equation}
One immediately checks that (\ref{GG}) is satisfied. Moreover, for $\nu \in \{0,-\frac{d}{d+2},-1\}$ the above inequality is saturated only for $d=2$.

\section{Dynamical maps beyond Markovian regime}

Consider now time-dependent generator (\ref{L}), i.e. $h_i=h_i(t)$, $\kappa=\kappa(t)$, and $\nu=\nu(t)$.
One finds

\[ \alpha(t) = 1 - A(t)\ , \quad \beta(t) = A(t) - B(t) ,  \]
with

\[   A(t) =  \exp\left(-d \int_0^t\kappa(\tau) {\rm d} \tau\right) \ , \ \ B(t) = \exp\left(-\int_0^t\kappa(\tau)[d-1+\nu(\tau)] {\rm d} \tau\right) .  \]
For $\kappa(t)=1$ and constant $\nu(t)=\nu$ one recovers (\ref{ab}). It is clear that $\Lambda_t= \exp\left( \int_0^t \mathcal{L}_\tau {\rm d}\tau\right)$ is CP-divisible iff $\kappa(t),\nu(t) \geq 0$ for all $t \geq 0$. However, $\Lambda_t$ may be completely positive for all $t \geq 0$ even if temporally $\nu(t) < 0$. Indeed, setting $\kappa(t)=1$ the condition for complete positivity of $\Lambda_t$ (cf. condition (\ref{CP})) `$- {\alpha(t)} \leq d\beta(t)$'  gives rise to the following condition for the parameter $\nu(t)$:

\[  N(t)\ge t-\ln\!\left(\frac{e^{dt}+d-1}{d}\right) ,\]
where $N(t) = \int_0^t \nu(\tau){\rm d}\tau$. Equivalently

\[ e^{-N(t)}\le \frac{e^{(d-1)t}+(d-1)e^{-t}}{d}. \]
The optimal choice (i.e. when the above inequality is saturated) leads to 

\begin{equation}  \label{n-t}
    \nu(t) = - (d-1) \frac{e^{dt} - 1}{e^{dt} + (d-1)} .
\end{equation}
Hence $\nu(t) < 0$ for all $t>0$ and $\nu(t) \to -(d-1)$ as $t \to \infty$. Note, that for $d=2$ we recover

\begin{equation}   \label{tanh}
    \nu(t) = -\frac{e^{2t}-1}{e^{2t}+1} = - \tanh\, t ,
\end{equation}
analyzed in \cite{Erica} (see also \cite{Nina,Jagadish}). One finds, therefore, the following formula for eternally non-Markovian evolution

\begin{equation}   \label{ENM}
\Lambda_t(X)
=
\frac{1+(d-1)e^{-dt}}{d}\,X
+\frac{e^{-dt}-1}{d}\,\Delta(X)
+\frac{1-e^{-dt}}{d}\,\oper_d\,\mathrm{Tr}X.
\end{equation}
Interestingly, the asymptotic map reads

\begin{equation}
    \Lambda_\infty(X) = \frac 1d \Big(X - \Delta(X) + \oper_d \Tr X \Big) = \mathcal{E}_4(X), 
\end{equation}
and hence this eternally non-Markovian evolution is eventually entanglement breaking and the asymptotic map $\mathcal{E}_4$ belongs now to the boundary of entanglement breaking channels, cf. Fig. \ref{Fig-4}. It shows that this evolution is no longer relaxing to the maximally mixed state but the asymptotic state does depend upon the initial state which is a clear sign of memory effects. It should be stressed that the above evolution is not P-divisible since the necessary condition for P-divisibility, that is, $\nu(t) \geq -1$, is violated. 

\begin{figure}[h!]
    \centering
    \includegraphics[width=11cm]{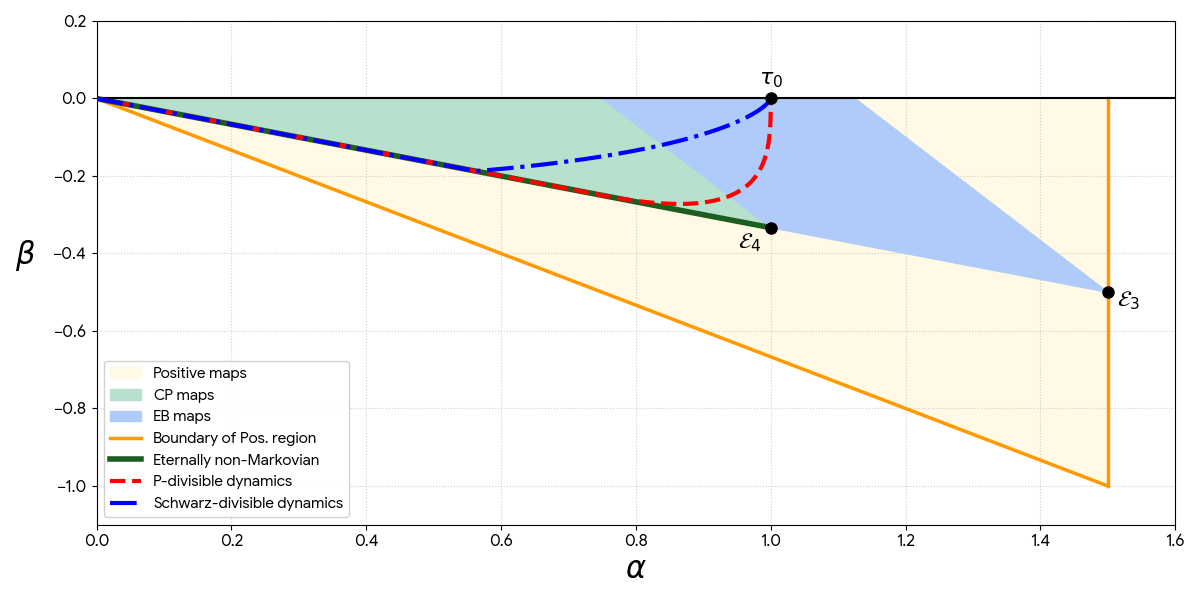} 
    \caption{Eternally non-Markovian evolution stays at the boudary of completely positive maps and approaches asymptotically EB channel $\mathcal{E}_4$. Schwarz-divisible is {\em less non-Markovian} than P-divisible. All three curves correspond to $d=3$.}    \label{Fig-4}
\end{figure}
To construct evolution which is P-divisible (but not CP-divisible) let us consider

\begin{equation}
    \nu(t) = \left\{ \begin{array}{ll} - (d-1) \frac{e^{dt} - 1}{e^{dt} + (d-1)} \ ,& \ t \leq t_* \\
    -1 \ ,& \ t > t_* \end{array} \right. \ , 
\end{equation}
with $t_*=\infty$ for $d=2$, and

\begin{equation} \label{t*}
    t_* = \frac 1d \ln\left( \frac{2(d-1)}{d-2} \right) , 
\end{equation}
for $d > 2$. By construction $\nu(t) \geq -1$ which guarantees P-divisibility. One easily checks that the corresponding $\Lambda_t$ is completely positive

\begin{equation*}
\Lambda_t(X)=
\begin{cases}
\dfrac{1+(d-1)e^{-dt}}{d}\,X
+\dfrac{e^{-dt}-1}{d}\,\Delta(X)
+\dfrac{1-e^{-dt}}{d}\,\oper_d\,\mathrm{Tr}X,
& 0\le t\le t_*,
\\[3ex]
\dfrac12 e^{-(d-2)(t-t_*)}\,X
+\left(e^{-dt}-\dfrac12 e^{-(d-2)(t-t_*)}\right)\Delta(X)
+\dfrac{1-e^{-dt}}{d}\,\oper_d\,\mathrm{Tr}X,
& t\ge t_*,
\end{cases}
\end{equation*}
and the asymptotic map $\Lambda_\infty(X) = \frac 1d \oper_d \Tr X$, i.e. P-divisible evolution relaxes to the maximally mixed state. 

Finally, consider dynamics which is not only P-divisible but also Schwarz-divisible (but non-Markovian):

\begin{equation}
    \nu(t) = \left\{ \begin{array}{ll} - (d-1) \frac{e^{dt} - 1}{e^{dt} + (d-1)} \ ,& \ t \leq t_S \\
    -\frac{d}{d+2} \ ,& \ t > t_S \end{array} \right. \ , 
\end{equation}
with 

\[ t_S = \frac 1d \ln\left( \frac{2(d^2-1)}{d^2-2} \right) . \]
Now, for $t > t_S$ one finds

\begin{equation*}
\Lambda_t(X)=
\dfrac{d+2}{2(d+1)}\,e^{-\frac{d^2}{d+2}(t-t_S)}\,X
+\left(e^{-dt}-\dfrac{d+2}{2(d+1)}\,e^{-\frac{d^2}{d+2}(t-t_S)}\right)\Delta(X)
+\dfrac{1-e^{-dt}}{d}\,\oper_d\,\operatorname{Tr}X .
\end{equation*}
Again, the asymptotic map is a completely depolarizing channel (cf. Fig. \ref{Fig-4}).

\begin{remark}
Note, that for  $d=2$ one finds $t_*=\infty$ and hence the map

\[ \Lambda_t(X)=
\dfrac{1+ e^{-2t}}{2}\,X
+\dfrac{1-e^{-2t}}{2}\,\Big( \oper_2\,\mathrm{Tr}X - \Delta(X) \Big) ,   \]
is P-divisible \cite{Erica}. The asymptotic state reads

\[  \Lambda_\infty(\rho) = \mathcal{E}_4(\rho) = \frac 12 \begin{pmatrix} 1 &  \rho_{12} \\  \rho_{21} & 1 
\end{pmatrix} ,\]
and it does depend on the initial state $\rho$. Moreover, Schwarz-divisible map is already CP-divisible and hence Markovian. It is no longer true for $d > 2$. 
\end{remark}
\begin{remark} For time-dependent $\nu(t)$ relaxation rates are also time-dependent:
\[  \Gamma_\ell = \kappa d \ , \ \ell=1,\ldots,d-1 ,\]
and
\[  \Gamma_{ij}(t) = \kappa [d-1+\nu(t)] \ , \ \ i \neq j . \]
For $\kappa=1$ and $\nu(t)$ defined in (\ref{n-t}) one finds asymptotically  ($t \to \infty$)

\[  \Gamma_\ell = d \ , \ \ \Gamma_{ij}=0 \ , \ \ \Gamma = d(d-1) ,\]
and hence the relation $\Gamma_\ell \leq \dfrac{1}{d} \Gamma$ is violated which is another signature of non-Markovianity.
\end{remark}

\begin{remark} Recall that the eternally non-Markovian dynamical map (\ref{ENM}) belongs to the boundary of the set of completey positive maps and approaches entanglement breaking channel $\mathcal{E}_4$. Consider now the following map

\begin{equation}   \label{ENM-2}
\Lambda_t(X)
= e^{-dt}\,X
+\frac{e^{-dt}-1}{d-1}\,\Delta(X)
+\frac{1-e^{-dt}}{d-1}\,\oper_d\,\mathrm{Tr}X.
\end{equation}
One easily checks that $\Lambda_t$ is CPTP for all $t \geq 0$ and it belongs to the boundary of the set of completely positive maps. One finds the corresponding time-local generator 
with time dependent parameters

\[ \kappa(t) = \frac{d}{d - e^{dt}} \ , \quad \nu(t) = 1 - e^{dt} . \]
Observe that  $\nu(t) < 0$ for $t > 0$  and $\kappa(t)$ is singular at $t_1 = \dfrac{\ln d}{d}$. 
This dynamics is also  eternally non-Markovian since for $t < t_1$ one has $(\kappa(t) > 0,\kappa(t) \nu(t)) < 0$ and for $t > t_1$ one has $(\kappa(t) < 0,\kappa(t)\nu(t) > 0)$. Interestingly,

\[ \Lambda_{t_1} = \mathcal{E}_4 \ , \quad \Lambda_\infty = \mathcal{E}_3 ,\]
that is, at singularity $t_1$ the map becomes entanglement breaking and stays EB for all $t > t_1$ approaching asymptotically EB channel $\mathcal{E}_3$.  For example for $d=2$ one finds

\[  \Lambda_\infty(\rho) = \begin{pmatrix}
    \rho_{11} & 0 \\ 0 & \rho_{00}
\end{pmatrix} ,\]
that is, $\Lambda_\infty(\rho) = \sigma_x \Delta(\rho) \sigma_x$ perfectly decoheres and flips the initial populations. 
\end{remark}

\begin{remark} Observe that for $d=2$ one has
 \[ \sigma_+ X \sigma_- + \sigma_- X \sigma_+ = \frac 12 (\sigma_x X \sigma_x + \sigma_y X \sigma_y) .\]
Now, introducing a set of Weyl matrices 

    \[   W_{k\ell} = Z^k X^\ell \ ; \ \ \ k,\ell=0,1,\ldots,d-1 , \]
by $Z = \sum_{k=0}^{d-1} \omega^k E_{kk}$, and $X|k\> = |k+1\>$, one finds

\[   \sum_{k,\ell=0}^{d-1} W_{k\ell} X W_{k\ell}^\dagger = d \,\oper_d \Tr X , \]
and hence the formula for $\mathcal{L}$ may be rewritten as follows

\begin{equation}   \label{LW}
    \mathcal{L}(X) =  \frac \kappa d \Big( \sum_k\sum_{\ell > 0} ( W_{k\ell} X W_{k\ell}^\dagger - X) + \nu \sum_k ( Z^k X Z^{k\dagger} - X) \Big) . 
\end{equation}
Let $\mathcal{L}_{k\ell}(X) = W_{k\ell} X W_{k\ell}^\dagger - X$ and consider the following CPTP map

\begin{equation}
    \Phi_t = \frac{1}{d(d-1)} \sum_k \sum_{\ell > 0} e^{t \mathcal{L}_{k\ell} } ,
\end{equation}
i.e. a convex combination of Markovian semigroups. The map $\Phi_t$ satisfies $\dot{\Phi}_t = \mathcal{L}_t \Phi_t$, where the time dependent generator $\mathcal{L}_t$ is defined by (\ref{LW}) with time dependent coefficient

\[ \kappa(t) = \frac{\dot f(t)}{f(t)}, \qquad \nu(t) = 1 - \frac{1}{f(t)} , \]
where

\[ f(t) = \frac{1}{d} \sum_{r=0}^{d-1} \exp\!\big(t(\omega^r - 1)\big), \qquad \omega = e^{2\pi i / d} .\]
For $d=2$ one finds $\kappa(t)=1$ and $\nu(t) = -\tanh t$, i.e. one recovers (\ref{tanh}). For $d>2$ the evolution still belongs to the boundary of completely positive maps (green triangle) but $\kappa(t)$ is no longer constant. Moreover, $\kappa(t)$ is highly singular and hence the dynamical map $\Phi_t$ is 
non-invertible \cite{PRL-Angel,Jagadish-2,Kasia}
    
\end{remark}

\section{Conclusions}

In this work, we have analyzed a class of quantum dynamical semigroups for 
$d$-level systems that interpolate between positive, Schwarz, and completely positive evolutions. By introducing a parametrized family of generators, we established clear criteria under which the resulting dynamical maps belong to each of these important classes. In particular, we showed how conditional positivity and complete positivity are governed by simple parameter constraints, allowing for a transparent characterization of the transition between different regimes of quantum dynamics.

A key result of the study is the geometric interpretation of these regions, illustrated through parameter spaces where positivity, complete positivity, and entanglement-breaking behavior can be directly compared (see, e.g., the diagram on page 6). This provides insight into how dynamical maps evolve over time and how physical properties such as entanglement are affected. Notably, we demonstrated that within the considered family, the evolution can naturally become entanglement breaking after a finite time, even when starting from a broader class of positive maps.

Furthermore, we examined relaxation rates and spectral properties of the generators, highlighting how these depend on the dimension and parameters of the system. Extensions beyond the Markovian regime were also discussed, showing that time-dependent generators can lead to non-Markovian dynamics while still preserving certain structural properties. Interestingly, we provided a family of examples of so-called eternally non-Markovian evolutions $\{\Lambda_t\}_{t\geq 0}$ such that for each $t \geq 0$ the map $\Lambda_t$ belongs to the border of the set of completely positive maps. Moreover, it was shown that properly engeneering the parameters $\kappa(t)$ and $\nu(t)$ one may realize dynamics with different asymptotics. 

Overall, the results contribute to a deeper understanding of the hierarchy of quantum maps and their dynamical behavior, offering a useful framework for studying open quantum systems and their transition between physically relevant regimes. It would be interesting to further analyze dynamics with $k$-positive or $k$-entangement breaking maps \cite{Franz,kEB} and to generalize the presented analysis to infinite dimensional scenario (see e.g. recent work \cite{Franz-2}).

\appendix

\section*{Appendix}

\section{Proof of Lemma \ref{L1}}

\begin{proof}
Setting $z_i:=x_i\overline{y_i}$ it is enough to show that
\[
\sum_{i=1}^d |z_i|^2\le \frac12.
\]
Since \(x\) and \(y\) are orthogonal, one has
\[
\sum_{i=1}^d z_i=\sum_{i=1}^d x_i\overline{y_i}=0 ,
\]
and hence
\[
0=\left|\sum_{i=1}^d z_i\right|^2
 =\sum_{i=1}^d |z_i|^2+2\sum_{1\le i<j\le d}\Re\!\bigl(z_i\overline{z_j}\bigr).
\]
Therefore,
\[
\sum_{i=1}^d |z_i|^2
=-2\sum_{1\le i<j\le d}\Re\!\bigl(z_i\overline{z_j}\bigr)
\le 2\sum_{1\le i<j\le d}|z_i||z_j|.
\]
Adding \(\sum_{i=1}^d |z_i|^2\) to both sides gives
\[
2\sum_{i=1}^d |z_i|^2
\le \sum_{i=1}^d |z_i|^2+2\sum_{1\le i<j\le d}|z_i||z_j|
=\left(\sum_{i=1}^d |z_i|\right)^2.
\]
Thus
\[
\sum_{i=1}^d |z_i|^2\le \frac12\left(\sum_{i=1}^d |z_i|\right)^2.
\]
Now, by the Cauchy--Schwarz inequality,
\[
\sum_{i=1}^d |z_i|
=\sum_{i=1}^d |x_i||y_i|
\le \left(\sum_{i=1}^d |x_i|^2\right)^{1/2}
   \left(\sum_{i=1}^d |y_i|^2\right)^{1/2}
=1.
\]
Consequently,
\[
\sum_{i=1}^d |z_i|^2\le \frac12.
\]
\end{proof}

\section{Proof of Proposition \ref{P-MAIN}}

Let us observe that $\mathcal{L}$ may be equivalently written as 

\begin{equation}
    \mathcal{L}(\rho) = -i [H,\rho] + \kappa \Big(\Psi_a(\rho) - (d-a)\rho\Big) ,
\end{equation}
where $a= 1-\nu$
and $\Psi_a : \Md \to \Md$ is a unital trace-preserving map

\begin{equation}
    \Psi_a(\rho) = \oper_d {\rm Tr}\rho - a\, \Delta(\rho) ,
\end{equation}
and $\Delta(\rho) = \sum_{k=1}^d E_{kk} \rho E_{kk}$. Hence, $\mathcal{L}$ is a dissipative generator iff $\Psi_a$ satisfies

\begin{equation}\label{Phi!!}
  \Psi_a(X^\dagger X) \geq \Psi_a(X^\dagger) X + X^\dagger \Psi_a(X) - (d-a) X^\dagger X,
\end{equation}
for all traceless $X \in \Md$. Equivalently, 
\begin{equation}\label{Ea}
M(a,X):=\tr(X^\dagger X)\oper_d +(d-a)X^\dagger X-a\Delta(X^\dagger X)+{a}\big(\Delta(X^\dagger)X+X^\dagger\Delta(X)\big)\geq 0 .
\end{equation}
For any $X\in M_d(\C)$ let
\[
X=D+N,\qquad D:=\Delta(X)={\rm Diag}(\delta_1,\dots,\delta_d),\quad N:=X-D ,
\]
be a decomposition of $X$ into diagonal part $D$ and off-diagonal part $N$. One easily finds

\begin{equation}\label{eq:pinch-key}
\Delta(X^\dagger X)=D^\dagger D+\Delta(N^\dagger N),
\qquad
X^\dagger\Delta(X)+\Delta(X^\dagger)X=2D^\dagger D+N^\dagger D+D^\dagger N.
\end{equation}
First we show that condition $a \leq 1 + \frac{d}{d+2}$ is necessary. Indeed, consider 
\[
X=\begin{pmatrix}1&-c\\ c&-1\end{pmatrix}\oplus 0_{d-2}\qquad(c\in\mathbb{R}),
\]
so $\tr X=0$ and $\D(X)= {\rm Diag}(1,-1,0,\dots,0)$.
A direct computation shows:
\[
X^\dagger X=\begin{pmatrix}1+c^2&-2c\\ -2c&1+c^2\end{pmatrix} \oplus 0_{d-2} ,\quad
\Delta(X^\dagger X)={\rm Diag}(1+c^2,1+c^2,0,\ldots,0) \]
and
\[
X^\dagger\Delta(X)+\Delta(X^\dagger)X=\begin{pmatrix}2&-2c\\ -2c&2\end{pmatrix} \oplus 0_{d-2} .
\]
Inserting into \eqref{Ea} one obtains 
\[ M(a,X) =
\begin{pmatrix}
d+2+(d+2-2a)c^2 & -\,2dc\\[2pt]
-\,2dc & d+2+(d+2-2a)c^2
\end{pmatrix} \oplus O_{d-2}
\]
The smallest eigenvalue of the $2\times 2$ block as a function of a parameter `$c$' reads
\[
\lambda(c)=d+2+(d+2-2a)c^2-2dc.
\]
If $d+2-2a>0$, the minimum over $c\in\mathbb{R}$ occurs at
$c^\star=\dfrac{d}{\,d+2-2a\,}$ and equals
\[
\lambda_{\min}
=d+2-\frac{d^2}{\,d+2-2a\,}
=\frac{2\bigl(2(d+1)-(d+2)a\bigr)}{\,d+2-2a\,}.
\]
Hence $\lambda_{\min}\geq 0$ only if  $a \leq \dfrac{2(d+1)}{d+2}$. Clearly, the upper bound for $a$ is tight. 

To prove that condition $a \leq \dfrac{2(d+1)}{d+2}$ is also sufficient let us observe 

\begin{equation}\label{eq:vMv}
\<y,M(X,a)y\>
=\|{X}\|^2\|y\|^2 + d\|Xy\|^2
-a\left(\|Ny\|^2+\sum_{k=1}^d |y_k|^2\,\|Ne_k\|^2\right) ,
\end{equation}
and hence assuming that  $y \in \mathbb{C}^d$ is normalized

\begin{eqnarray}
    \ip{y}{M(a,X)\,y}
&=& \|D\|_{HS}^2+\|N\|_{HS}^2
  +d\,\|D y\|^2
  +2d\,\re\ip{Ny}{D y} \nonumber \\
&+& (d-a)\|Ny\|^2
  -a\sum_{i=1}^d p_i\,\|\mathrm{col}_i(N)\|^2,   \label{yEy}
\end{eqnarray}
where $p_i = |y_i|^2$ and $\mathrm{col}_i(N)$ stands for $i$th column of the matrix $N$. Indeed,  (\ref{yEy}) follows from direct computation
\begin{align*}
\ip{y}{\Tr(X^\dagger X)I\,y}&= \|D\|_{HS}^2+\|N\|_{HS}^2,\\
\ip{y}{X^\dagger X\,y}&= \|D y\|^2+\|Ny\|^2+2\,\re\ip{Ny}{D y},\\
\ip{y}{\D(X^\dagger X)\,y}&=\|D y\|^2+\sum_i p_i\,\|\mathrm{col}_i(N)\|^2,\\
\ip{y}{X^\dagger\D(X)+\D(X^\dagger)X\,y}&=2\|D y\|^2+2\,\re\ip{Ny}{D y}.
\end{align*}
Now, minimizing the r.h.s. of (\ref{yEy}) w.r.t. the off-diagonal part $N$ one arrives at the following 

\begin{lemma}\label{lem:rowwise}
For any normalized vector $y \in \mathbb{C}^d$
\begin{equation}\label{eq:bound17}
\langle y,M(a,X)y\rangle
\ge
\sum_{i=1}^d |\delta_i|^2\left(
1 + d p_i - \frac{d^2 p_i S_i}{1+(d-a)S_i}
\right),
\end{equation}
where $p_i = |y_i|^2$ and 
\begin{equation}
S_i:=\sum_{j\neq i}\frac{p_j}{1-a p_j}.
\end{equation}
\end{lemma}
\begin{proof}
    Fix a row index $i$ and set
$u_j:=N_{ij}y_j$ for $j\ne i$, and $v_i:=\sum_{j\ne i}u_j=(Ny)_i$.
The contribution from the matrix $N$  to the r.h.s. of \eqref{yEy} reads
\[
n_i = \sum_{j\ne i}\Bigl(\tfrac{1}{p_j}-a\Bigr)|u_j|^2
\;+\;(d-a)|v_i|^2
\;+\;2d\,\re\!\bigl(v_i\,\overline{\delta_i y_i}\bigr).
\]
For fixed $v_i$, the minimum of $\sum_{j\ne i}(\frac{1}{p_j}-a)|u_j|^2$
subject to $\sum_j u_j=v_i$ (Lagrange multiplier method) reads $|v_i|^2/S_i$, with
 $S_i=\sum_{j\ne i}\frac{p_j}{1-a p_j}$.
Minimizing the formula for $n_i$ over $v_i\in\C$ then gives
\[ n_i = 
-\,\frac{d^2\,|\delta_i y_i|^2\,S_i}{\,1+(d-a)S_i\,}
=-\,\frac{d^2\,p_i\,|\delta_i|^2\,S_i}{\,1+(d-a)S_i\,}.
\]
Summing over $i=1,\ldots,d$ finally proves \eqref{eq:bound17}.
\end{proof}

\begin{lemma}[Two--coordinate reduction]\label{lem:two_coord} 
A supremum of $S_i$ over all
$(p_j)_{j\neq i}$ with fixed sum $T_i := \sum_{j\neq i}p_j=1-p_i$ is attained on a vector $(p_i)$ with at most
two non-zero coordinates $p_j$ and $p_k$, with  $j,k\neq i$:

\begin{itemize}
\item if $T_i\le 1/a$, the maximum is attained when $p_j=T_i$ for a single $j\neq i$ and
      $p_\ell=0$ for all $\ell\neq i,j$;
\item if $T> 1/a$, then $\sup S_i=+\infty$, and one can make $S_i$ arbitrarily large by taking
      $p_j\to(1/a)^-$ and $p_k=T_i-p_j$, with all other $p_\ell=0$.
\end{itemize}
Consequently, since $S\mapsto \frac{S}{1+(d-a)S}$ is increasing on $[0,+\infty]$, to minimize the
right-hand side of \eqref{eq:bound17} over $y$ with $\|y\|^2=1$ it suffices to consider vectors $y$
supported on two coordinates.
\end{lemma}

\begin{proof}
Let $\varphi(t):=\frac{t}{1-a t}$. On $[0,1/a)$ we have
\[
\varphi'(t)=\frac{1}{(1-a t)^2}>0,\qquad
\varphi''(t)=\frac{2a}{(1-a t)^3}>0,
\]
so $\varphi$ is increasing and convex on that interval, and has a vertical asymptote at $t=1/a$.
If $T_i\le 1/a$, then $\sum_{j\neq i}\varphi(p_j)$ is maximized over the simplex
$\{(p_j)_{j\neq i}\ge0:\sum_{j\neq i}p_j=T_i\}$ at an extreme point, i.e.\ when one coordinate equals $T_i$
and the rest are $0$. If $T>1/a$, we may take one coordinate $p_j\to(1/a)^-$ while keeping
$\sum_{j\neq i}p_j=T_i$, and then $\varphi(p_j)\to+\infty$, hence $\sup S_i=+\infty$.
The final claim follows because in \eqref{eq:bound17} each term is decreasing in $S_i$ through the increasing
factor $\frac{S_i}{1+(d-a)S_i}$.
\end{proof}
We may assume $a<2$ (this holds in the range $a\le \frac{2(d+1)}{d+2}<2$ of the theorem).
Fix $X=\Delta+N$ and $y\in\mathbb C^d$, let $\alpha=\|y\|^2$, and apply the lower bound \eqref{eq:bound17}.
By Lemma~\ref{lem:two_coord}, it suffices to take $y$ supported on two indices, say $\{r,s\}$.
By homogeneity in $y$ we set $\alpha=1$ and write
\[
p:=p_r\in[0,1],\qquad p_s=1-p,\qquad t:=p(1-p)\in\Bigl[0,\frac14\Bigr].
\]

For $i\notin\{r,s\}$ we have $p_i=0$, hence the $i$--term in \eqref{eq:bound17} equals $|\delta_i|^2$.
For $i=r,s$ we use the extremal $S_i$ from Lemma~\ref{lem:two_coord}, namely
\[
S_r=\frac{1-p}{1-a(1-p)},\qquad S_s=\frac{p}{1-a p}.
\]
A direct simplification gives
\[
\frac{p\,S_r}{1+(d-a)S_r}=\frac{t}{1+(d-2a)(1-p)},\qquad
\frac{(1-p)\,S_s}{1+(d-a)S_s}=\frac{t}{1+(d-2a)p}.
\]
Therefore
\begin{equation}\label{eq:F_t}
\frac{p\,S_r}{1+(d-a)S_r}+\frac{(1-p)\,S_s}{1+(d-a)S_s}
=\frac{(d+2-2a)\,t}{(d+1-2a)+(2a-d)^2\,t}.
\end{equation}
Inserting into \eqref{eq:bound17} yields
\begin{eqnarray}\label{delta_rs}
\langle y,M(a,X)y\rangle
&\geq&
|\delta_r|^2\Bigl(1+d p\Bigr)+|\delta_s|^2\Bigl(1+d(1-p)\Bigr)\nonumber \\[2mm]
&&-d^2 t\Biggl(\frac{|\delta_r|^2}{1+(d-2a)(1-p)}
+\frac{|\delta_s|^2}{1+(d-2a)p}\Biggr)
+\sum_{i\neq r,s}|\delta_i|^2.
\end{eqnarray}

Since $\Tr X=0$ and $\Tr N=0$ (because $\Delta N=0$), we have $\sum_{i=1}^d \delta_i=0$.
If $d=2$ then necessarily $\delta_s=-\delta_r$ and the argument below reduces to the $2\times2$ case.
Assume $d\ge3$. By Cauchy--Schwarz,
\[
\sum_{i\neq r,s}|\delta_i|^2\ \ge\ \frac1{d-2}\Bigl|\sum_{i\neq r,s}\delta_i\Bigr|^2
=\frac{|\delta_r+\delta_s|^2}{d-2}.
\]
Using this in \eqref{delta_rs}, we obtain a $2\times 2$ Hermitian quadratic form in $(\delta_r,\delta_s)$:
\[
\langle y,M(a,X)y\rangle
\ \ge\
\begin{bmatrix}\delta_r&\delta_s\end{bmatrix}
M(p)\begin{bmatrix}\bar \delta_r\\ \bar \delta_s\end{bmatrix},
\]
where
\[
M(p):=
\begin{pmatrix}
c_r(p)+\frac1{d-2} & \frac1{d-2}\\[0.6ex]
\frac1{d-2} & c_s(p)+\frac1{d-2}
\end{pmatrix},
\]
and
\[
c_r(p):=1+d p-d^2\frac{p\,S_r}{1+(d-a)S_r},\qquad
c_s(p):=1+d(1-p)-d^2\frac{(1-p)\,S_s}{1+(d-a)S_s}.
\]
It is, therefore, clear that to complete the proof it is sufficient to prove the following 
\begin{lemma}   \label{lemma-F}
    If $a\leq \frac{2(d+1)}{d+2}$, then $M(a,X) \geq 0$.
\end{lemma}
\begin{proof}
    Since $M(p)$ is $2\times2$ Hermitian, it is positive semidefinite iff $\Tr M(p)\ge0$ and $\det M(p)\ge0$.

First, using \eqref{eq:F_t} we get
\begin{eqnarray*}
\Tr M(p)&=&c_r(p)+c_s(p)+\frac{2}{d-2}\\[2mm]
&=&2+d-d^2\cdot\frac{(d+2-2a)\,t}{(d+1-2a)+(2a-d)^2 t}+\frac{2}{d-2}.
\end{eqnarray*}
The map $t\mapsto \frac{t}{A+Bt}$ is increasing on $[0,\frac14]$ whenever $A\ge0$ and $B\ge0$.
Here $A=d+1-2a\ge0$ in the range $a\le \frac{2(d+1)}{d+2}$ and $B=(2a-d)^2\ge0$.
Hence the negative term is maximized at $t=\frac14$, so
\begin{eqnarray*}
c_r(p)+c_s(p)\ &\geq&\ 2+d-\frac{d^2}{d+2-2a}\\[2mm]
&=&\frac{2\bigl(2(d+1)-(d+2)a\bigr)}{d+2-2a}\\
&\geq& 0,
\end{eqnarray*}
and therefore $\Tr M(p)\ge \frac{2}{d-2}\ge0$.

Second, one computes (using $t=p(1-p)$) the determinant in closed form:
\[
\det M(p)
=
\frac{d^2\,(1-2a t)\bigl((d+1-2a)-2a(d-2)t\bigr)}
{(d-2)\bigl((d+1-2a)+(2a-d)^2 t\bigr)}.
\]
For $t\in[0,\frac14]$ and $a<2$ we have $1-2a t\ge 1-\frac a2>0$.
Moreover, the factor $(d+1-2a)-2a(d-2)t$ is minimized at $t=\frac14$, where it equals
\[
(d+1-2a)-\frac{a}{2}(d-2)=\frac{2(d+1)-(d+2)a}{2}\ \ge\ 0
\quad\text{for }\ a\le \frac{2(d+1)}{d+2}.
\]
Finally, the denominator is positive on $[0,\frac14]$ when $d+1-2a\ge0$.
Thus $\det M(p)\ge0$ in the stated range.

Hence $M(p)\geq 0$, so $\langle y,M(a,X)y\rangle\ge0$ for all $y$ and all $X$ with $\Tr X=0$.
This proves $M(a,X)\geq 0$ for $a\le \frac{2(d+1)}{d+2}$.

\end{proof}

\section{Proof of Proposition \ref{Pro-EB}}

The Choi matrix of $\Phi_{\alpha,\beta}$ belongs to a class of bipartite operators 

\[
\rho_{p,q}:=(1-p-q)P^+ + p\,\oper_d\otimes \oper_d + q\,D,
\qquad
D:=\sum_{k=1}^d E_{kk}\otimes E_{kk},
\]
where
\[
P^+=|\Omega\rangle\langle\Omega|,
\qquad
|\Omega\rangle=\frac{1}{\sqrt d}\sum_{i=1}^d |ii\rangle.
\]
One has that $\rho_{p,q} \geq 0$ iff

\[
p\ge 0
\qquad\text{and}\qquad
p+q\ge 0.
\]

\begin{lemma} Let $\rho_{p,q} \geq 0$. $\rho_{p,q}$ is separable if and only if it is PPT.
\end{lemma}
\begin{proof}
    Simple analysis shows that the operator $\rho_{p,q}$ is PPT if and only if
\[
p\ge 0,\qquad p+q\ge 0,\qquad
1-(d+1)p \le q \le 1+(d-1)p.
\]
Equivalently,
\[
p-\frac{1-p-q}{d}\ge 0
\quad\text{and}\quad
p+q+\frac{1-p-q}{d}\ge 0.
\]
Now, we show that any PPT state $\rho_{p,q}$ is separable. 
Defining $Q = \oper_d \otimes \oper_d - D$, one has

\[ \rho_{p,q} = t P^+ + p \,Q + (p+q) D = \frac{t}{d} S + \left( p - \frac td\right) Q + (p+q) D , \]
where $t=1 -p -q$, and 

\[ S = Q - dP^+ . \]
Now, observe that 

\[ S = d^2 \< |\phi(\bm{x})\>\<\phi(\bm{x})| \otimes |\overline{\phi(\bm{x})}\>\<\overline{\phi(\bm{x})}| \> \]
where

\[ |\phi(\bm{x})\> = \sum_{k=1}^d e^{i x_k} |e_k\> ,  \]
with $\bm{x} = (x_1,\ldots,x_d) \in \mathbb{R}^d$. Finally $\< ... \>$ denotes an average over vector of phases $\bm{x}$. This proves that $S$ is a separable operator. Hence, if $\rho_{p,q}$ is PPT it is necessarily separable being a convex combination of separable operators. 

\end{proof}

\section*{Acknowledgements}

D.C. was supported by the Polish National Science
Centre project No. 2024/55/B/ST2/01781.

\end{document}